\documentclass[11pt]{article}
\usepackage{hyperref}
\normalsize
\usepackage{amsfonts, amsthm, amsgen}

\usepackage[latin1]{inputenc}

\usepackage{authblk}
\usepackage{amssymb,latexsym,amsmath}
\usepackage{amsthm}

\newtheorem{theorem}{Theorem}[section]
\newtheorem{proposition}[theorem]{Proposition}

\newtheorem{lemma}[theorem]{Lemma}
\newtheorem{corollary}[theorem]{Corollary}

\newtheorem{example}[theorem]{Example}

\newtheorem{remark}[theorem]{Remark}

\newtheorem{definition}[theorem]{Definition}

\newcommand{\F}{\mathbb{F}}
\newcommand{\N}{\mathbb{N}}
\newcommand{\Z}{\mathbb{Z}}

\newcommand{\gen}{\mathrm{gen}}

\newcommand{\mC}{\mathcal{C}}

\newcommand{\Fq}{\mathbb{F}_{q}}

\newcommand{\Fqm}{\mathbb{F}_{q^{m}}}
\newcommand{\dual}{C^{\perp}}

\newcommand{\dc}{\mathrm{def}\left(C\right)}
\newcommand{\defect}{\mathrm{def}}
\newcommand{\chara}{\mathrm{char}}
\newcommand{\Fqmn}{\mathbb{F}_{q^{m}}^{n}}
\newcommand{\MrC}{\mathcal{M}_{r}\left(C\right)}

\newcommand{\m}{\mathcal{M}}

\newcommand{\mA}{\mathcal{A}}

\newcommand{\mM}{\mathcal{M}}

\newcommand{\dist}{\mbox{\rm d}}
\newcommand{\MRD}{\mbox{\rm MRD}}

\DeclareMathOperator{\rank}{rank}

\DeclareMathOperator{\rdef}{Rdef}

\title{\bf On dually almost MRD codes }

\author{ Javier de la Cruz\footnote{
This work was done while J. de la Cruz was at the University of Zurich supported by the Swiss Confederation through the Swiss Government Excellence Scholarship no. 2016.0873. The autor was partially supported by COLCIENCIAS through project no. 121571250178.} \\Universidad del Norte, Barranquilla, Colombia\\  and \\ University of Zurich,  Switzerland}

\begin{document}

\maketitle

\begin{abstract}  In this paper we define and study a family of codes which come close to be MRD codes, so we call them AMRD codes (almost MRD).
 An AMRD code is a code with rank defect equal to 1. AMRD codes whose duals are AMRD are called dually AMRD. Dually AMRD codes are the closest to
the MRD codes given that both they and their dual codes are almost optimal. Necessary and sufficient conditions for the codes to be dually AMRD are given.
Furthermore we show that dually AMRD codes and  codes of rank defect one and maximum 2-generalized weight coincide when the size of the matrix divides the dimension.
  \end{abstract}

\section{Introduction}

Rank metric codes have cryptographic applications and applications in tape recording. Recently it was
shown how to employ them for error correction in coherent linear network coding (\cite{KK1}, \cite{metrics} \cite{metrics3}).
Due to these applications, there is a steady stream of work that focuses on general properties of codes with rank metric.

There exist two representations of rank metric codes: matrix representation and vector representation. In matrix representation linear rank metric codes are $\F_q$-linear subspaces of $ (\F_q)_{n,m}$, where the norm of an element $A\in  (\F_q)_{n,m}$ is defined as the rank of the matrix. In vector representation, rank metric codes are $\Fqm$-linear subspaces of the vector space $\Fqmn$, where the norm of a vector $v\in \Fqmn$ is defined as the maximal number of coordinates of $v$ which are linearly independent over $\Fq$.

An MRD code is a rank metric code which is maximal in size given the minimum
distance, or in other words it achieves the Singleton bound for the rank metric distance.
Delsarte \cite{Delsarte} and independently Gabidulin \cite{Gabidulin} proved the existence of $\F_q$-linear MRD codes for all $q,m, n$ and dimension
$1\leq t \leq mn$ divisible by $m$. Given the parameters $q,m, n, k$, the code $C\leq \Fqmn$, they describe has a particular construction through a generator matrix $M_k(v)$ and is called in the literature the Gabidulin code $\mathcal{G}$. Recently new constructions of MRD codes have been found which are not equivalent to Gabidulin codes $\mathcal{G}$ (\cite{Willems-DelaCruz-K-W,sheekey}).

In analogy with the Singleton defect for classical
codes, in \cite{DC-G-R-L} the authors propose a definition of rank defect for $\F_q$-linear rank metric codes. The rank defect
of a code $\mC\leq (\F_q)_{n,m}$ measures how far $\mC$ is away from being a MRD code. Based on this
concept a QMRD code is defined in \cite{DC-G-R-L} as an $\F_q$-linear code $\mC\leq (\F_q)_{n,m}$ with rank defect 0 and which is not MRD, i.e $\mC$ has rank defect 0 and $m \nmid t$.

In this work we define and study a family of codes which come close to be MRD codes, called dually almost MRD codes or simply dually AMRD codes.
This paper is structured as follows. In Section 2 we give the preliminaries on rank metric codes, rank defect, QMRD codes, rank distribution and generalized weights.
In Section 3 we present the definition of dually AMRD codes, we give necessary and sufficient conditions for its existence based on the parameters $n,m,t$ and $d$, moreover we establish its existence for the case $m=t$.

Using the rank distribution in Section 4 we give sufficient and necessary conditions for the code to be dually AMRD.
In particular we establish a relationships between the number of vectors of minimal weight of the original code $\mC$ and its dual code $\mC^\perp$,
which guarantees that the code is dually AMRD. We also analyze the self-dual AMRD codes.

Finally, in Section 5 we study the generalized weights for $\F_q$-linear codes. In this part we establish relations between the generalized weights of an $\F_q$-linear code
and its rank defect. We also prove that, when $m$ divides the dimension $t$, the concept of dually AMRD for $\F_q$-linear codes coincides with the concept of 2-AMRD for $\F_q$-linear codes (or near almost codes).

\section{Preliminaries}
Let $\F_q$ denote a finite field with $q$ elements and let $V=(\F_q)_{n,m}$ be the $\F_q$-vector space of
matrices over $\F_q$ of type $(n,m)$. On $V$ we define the so-called rank metric distance by $ \dist(A,B) =\rank(A-B) $
for $A,B \in V$.

A $t$-dimensional $\F_q$-subspace $\mathcal{\mC} \leq V$ endowed with the metric $\dist$ is called a {\it $\F_q$-linear rank metric} code with minimum distance $\dist(\mathcal{C})= \min\,\{ \dist(A,B)\mid A \not= B \in \mC \}.$
Clearly, the minimum distance of a code $\mC \neq \{ 0\}$ is also
$$d(\mC):=\min \{ \rank(A): A \in \mC, \ A\neq 0 \}.$$
Similarly, as in the classical coding theory, the rank distribution of $\mC$ is the collection ${(A_i(\mC))}_{i \in \N}$,
where $A_i(\mC):=|\{ A\in \mC: \rank(A)=i\}|$ for $i \in \N$.
The dual of a code $\mC \leq V$ is the code
$$\mC^\perp:=\{ B \in (\F_q)_{n,m} : \mbox{Tr}(BA^t)=0 \mbox{
for all } A \in \mC\}.$$
A code $\mC\leq V$ is \textit{self-dual} provided $\mC=\mC^\perp.$ As $\dim_{\F_q}(\mC)+\dim_{\F_q}(\mC^\perp)=\dim_{\F_q} V=nm$ a self-dual code $\mC$ has dimension $\frac{nm}{2}.$

 The field $\F_{q^m}$ may be viewed as an $m$-dimensional vector space over $\F_q$. The
{\it rank} of a vector
$v=(v_1,\dots,v_n) \in \F_{q^m}^n$ is defined as the maximum number of coordinates in $v$ that are linearly independent over $\F_q$, i.e.
$\rank(v):= \dim_{\F_q}\langle v_1,\dots,v_n \rangle$. Then we have a rank metric distance given by $ \dist(v,u) =\rank(v-u) $
for $v,u \in \F_{q^m}^n$.
An $\F_{q^m}$-linear subspace $C \leq \F_{q^m}^n$ of dimension $k$ endowed with this metric is called an {\it $\Fqm$-linear rank-metric} $[n,k]$ code.
The minimum distance
of a code $C \neq \{ 0\}$ is $$d(C):=\min \{ \rank(v): v \in C, \ v \neq 0
\}.$$
A code $C \leq \F_{q^m}^n$ is \textit{self-dual} provided $C=C^\perp$, where $C^\perp$ is defined with respect to the standard inner product of $\Fqmn.$

Let $v=(v_1,\dots,v_n) \in \F_{q^m}^n$, and let $\mathcal{B}=\{ \gamma_1,\dots,\gamma_m\}$
be a basis of $\F_{q^m}$ over $\F_q$. The {\it $\F_q$-linear  code associated}
 to an $\Fqm$-linear code $C \leq \F_{q^m}^n$ with respect to the basis $\mathcal{B}$ is
$$\lambda_{\mathcal{B}}(C):= \{ \lambda(v): v\in C\},$$ where $ \lambda(v)=(\lambda_{i,j})\in(\F_q)_{n, m}$ is the matrix such that $v_i= \sum_{j=1}^m \lambda_{i, j} \gamma_j$ for all $i=1, \ldots, n$.

It is well known that the rank distributions of $C$ and $\lambda_{\mathcal{B}}(C)$ agree and $\dim_{\F{_q}}(\lambda_{\mathcal{B}}(C))=m \cdot \dim_{\F_{q^m}}(C).$ In general $ C^\perp \neq \lambda_{\mathcal{B}}(C)^\perp $, however it has been shown in~\cite{Alberto} that their rank distributions also agree.

\begin{theorem}[\cite{Delsarte}, Theorem 5.6] \textsc{(Singleton-bound)} \label{singleton-bound for Delsarte codes}
Let $\mC\leq(\F_q)_{n,m}$ be an $\F_q$-linear code of dimension $t$ with minimum distance $d$. Then we have
$$d \leq \min\{ n-t/m+1, m-t/m+1\}.$$

In particular if $C \leq \Fqmn$ is an $\F_{q^m}$-linear code of dimension $k$, then $$d\leq \min\{n-k+1, \frac{m}{n}(n-k)+1 \}.$$
\end{theorem}

Rank metric codes meeting the Singleton bound are called {\it Maximum Rank Distance} (MRD) codes. Delsarte was the first who proved in \cite{Delsarte} the existence of linear $\MRD$ codes.

Given a vector $v=(v_1, \ldots, v_n) \in \Fqmn$ we denote by $M_k(v) \in (\F_{q^m})_{k,n}$ the matrix
$$ M_k(v)=\left(
            \begin{array}{cccc}
              v_1 & v_2 & \ldots & v_n \\
              v_1^{[1]} & v_2^{[1]} & \ldots & v_n^{[1]} \\
               &  & \vdots &  \\
              v_1^{[k-1]} & v_2^{[k-1]} & \ldots & v_n^{[k-1]} \\
            \end{array}
          \right),
$$ where $[i]:=q^i.$

Gabidulin showed in \cite{Gabidulin} that if $v_1, \ldots, v_n$ are linearly independent over $\F_q$,
then the $\Fqm$-linear code $C\leq \Fqmn$ generated by the matrix $M_k(v_1, \ldots, v_n)$ is a $k$-dimensional MRD code and we call it the {\it Gabidulin code $\mathcal{G}$ generated by $M_k(v_1, \ldots, v_n)$.}

\begin{remark} {\rm Throughout the paper, $d$ and $d^\perp$ denote the minimum distance of the code and its dual respectively. Furthermore, in this work we assume $n\leq m$, therefore $d \leq n-t/m+1 $ for an $\F_q$-linear code of dimension $t$ and $d \leq n-k+1$ for an $\Fqm$-linear code of dimension $k$. Also, unless stated otherwise, we will only consider non-trivial codes, i.e. $\{0\}\neq \mC \neq  (\F_q)_{n,m}$ and $\{0\}\neq C \neq  \Fqm^n$.}

\end{remark}

In analogy with the Singleton defect for classical codes given in \cite{Faldum-Willems},
we have the following definition for the rank defect of rank metric codes.

\begin{definition}[\cite{DC-G-R-L}]\label{rkdef-Delsarte}  {\rm
The {\it rank defect} of an $\F_q$-linear code $\mC \leq (\F_q)_{n,m}$ is defined by
$$\rdef(\mC)=n-\left\lceil\frac{t}{m}\right\rceil+1-d.$$
 If $C \leq \Fqmn$ is an $\Fqm$-linear $[n,k, d]$ code, then the rank defect of $C$ is defined as the defect of the associated code $\lambda(C)$,  i.e.
$\rdef(C)=n-k+1-d$.
}

\end{definition}
Note that $\rdef(\mC)=0$ if $\mC$ is MRD. However, $\rdef(\mC)$ may be zero also for codes $\mC$ which are not MRD. These codes are the closest codes to the MRD codes and are called quasi MRD codes (see \cite{DC-G-R-L}). Concretely, we have the following definition.

\begin{definition}[\cite{DC-G-R-L}]  {\rm
A code $\mC$ of dimension $t$ is {\it Quasi-MRD}, or {\it QMRD}, if $m\nmid t$ and $\rdef(\mC)=0$.}
\end{definition}
In \cite{DC-G-R-L} it is shown that as for MRD codes  QMRD codes exist for all choices of the parameters $1\leq n\leq m$ and $1\leq t<nm$ such that $m\nmid t$.

Generalized weights for $\Fqm$-linear codes were introduced in \cite{Kurihara-Matsumoto-Uyematsu, Ogg}. It was proved in \cite{Ducoat} that, through a refinement,
the definition given in \cite{Ogg} agrees with the definition of \cite{Kurihara-Matsumoto-Uyematsu}.
Similarly as V.K. Wei, who in \cite{Wei} studied the generalized weights for codes with Hamming metric motivated by cryptographical
applications in the wire-tap channel of type II, the authors in \cite{Kurihara-Matsumoto-Uyematsu, Ogg} introduce generalized rank weights to study the equivocation of wire-tap codes for network coding.

\begin{definition}[\cite{Kurihara-Matsumoto-Uyematsu}] {\rm
Given an $\Fqm$-linear code $C \leq \Fqmn$ of dimension $k$ and an integer $1 \leq r \leq k$ the $r^{th}$-\textit{generalized weight} of $C$ is
$$\MrC:=\min\{\dim_{\F_{q^m}}(V): V\in \Gamma(\Fqmn), \dim_{\F_{q^m}}(V\cap C)\geq r\},$$ where
$\Gamma(\Fqmn):=\{ V \leq \Fqmn : V^{q}=V\}$ and  $V^{q}:=\{ v^{q}:=(v_{1}^{q},\ldots v_{n}^{q}):v\in V\}$.
}
\end{definition}
The following theorem summarizes the main properties of the generalized weights for $\Fqm$-linear codes, which are similar to the generalized weights for
codes with the Hamming metric given by V.K. Wei in \cite{Wei}.
\begin{theorem}[\cite{Ducoat}, \cite{Kurihara-Matsumoto-Uyematsu}]\label{Ducoat}
Let $C \leq \Fqmn$ be an $\Fqm$-linear code of dimension $k$. Then we have
\begin{enumerate}
\item $\mathcal{M}_1(C)=d(C)$.
\item $\mathcal{M}_k (C) \leq n$.
\item For any $1\leq r \leq k$, we have $\mathcal{M}_r(C) < \mathcal{M}_{r+1}(C).$
\item For every $1\leq r\leq k$, we have $\MrC\leq n-k+r.$
\item $\{\mathcal{M}_1(C), \ldots, \mathcal{M}_k(C)\}=[n]\backslash \{ n+1-\mathcal{M}_{n-k}(\dual), n+1-\mathcal{M}_{1}(\dual)\},$ where\\ $[n]:=\{1, \ldots, n\}.$

\end{enumerate}
\end{theorem}

For $\F_q$-linear codes the generalized weights were introduced in \cite{a2}, refining previous definitions for $\Fqm$-linear codes given
in \cite{Kurihara-Matsumoto-Uyematsu, Ogg,Ducoat} and considering an anticode approach.

\begin{definition}[\cite{a2}] {\rm
An {\it optimal anticode} $\mA \leq (\F_q)_{n,m}$ is an $\F_q$-linear code such that
$\dim(\mA)=m \cdot \mbox{maxrk}(\mA)$, where
$\mbox{maxrk}(\mA):= \max \{ \mbox{rk}(M) : M \in \mA\}$.

Given an $\F_q$-linear code $\mC$ of dimension $t$ and an integer $1 \le r \le t$, the {\it $r$-th generalized weight}
of $\mC$ is
$$a_r(\mC):= \min \{ \mbox{maxrk}(\mA) : \mA \subseteq (\F_q)_{n,m} \mbox{ is an optimal anticode with
} \dim(\mC \cap \mA) \ge r\}.$$}
\end{definition}

\begin{theorem}[\cite{a2}] \label{theorem 30} Let $\mC \leq (\F_q)_{n,m}$ be an $\F_q$-linear code of dimension $t$. The following hold:
\begin{enumerate}
\item $a_1(\mC)=d(\mC)$.
\item $a_t(\mC)\leq n$.
\item For any $1 \leq r \leq t-1$, we have $a_r(\mC)\leq a_{r+1}(\mC)$.
\item For any $1 \leq r \leq t-m$, we have $a_r(\mC)< a_{r+m}(\mC)$
\item For any $1 \leq r \leq t$, we have $a_r(\mC)\leq n-\lfloor \frac{t-r}{m}\rfloor$.
\end{enumerate}
\end{theorem}
The following theorem, which was proven in \cite{a2}, shows that for $\F_q$-linear codes the generalized weights refine, as an algebraic invariant, generalized rank weights of $\Fqm$-linear codes.
\begin{theorem}[\cite{a2}] \label{corollary 31}
Let $C \leq \Fqmn$ be an $\Fqm$-linear code of dimension $k$. For any basis $\mathcal{B}$ of $\F_{q^m}$ over $\F_q$ and for any integers $1 \leq r \leq t$ and $0 \leq \epsilon \leq m-1$, we have $m_r(C)=a_{rm-\epsilon}(\lambda_{\mathcal{B}}(C))$.
\end{theorem}
\begin{remark} {\rm We often write $a_r$ or $a^\perp_r$ to denote $a_r(\mC)$ or $a_r(\mC^\perp)$ respectively. }

\end{remark}
\begin{theorem}[\cite{a2}] \label{theorem 32} Let $\mC \leq (\F_q)_{n,m}$ be an $\F_q$-linear code of dimension $t$ and $\mC^\perp$ its dual code. Then we have
$$\{n+1-a_{1+t-m}(\mC), \ldots, n+1-a_{1+t-\lfloor t/m\rfloor m}(\mC)\}=[n] \backslash \{ a_1(\mC^\perp), \ldots, a_{1+(n-\lceil \frac{t+1}{m}\rceil) m}(\mC^\perp)  \},$$ where $[n]:=\{1, \ldots, n\}.$
\end{theorem}

%

\section{Dually almost MRD codes}
We want to study the class of codes which are close to being MRD. In analogy with the definition of almost MDS codes (see \cite{Faldum-Willems}), we define:

\begin{definition} {\rm The code $\mC$ is an $s$-almost MRD code or {A$^s$MRD} if and only if $\rdef(\mC)=s$.
A$^1$MRD codes are simply called {AMRD} codes. Equivalently a code $\mC$ is an AMRD code if and only if $d=n-\lceil t/m\rceil$.}

\end{definition}
It is known that MRD and QMRD $\F_q$-linear codes exist for all $m, n, t, q$. Next we see that AMRD codes exist also for these parameters in the case $m \mid t$.

\begin{lemma}[Existence of AMRD codes] If $\mathcal{G} \leq \Fqmn$ is the Gabidulin code generated by $M_k(v_1, \ldots, v_n)$, then the extended code $\widehat{\mathcal{G}}$
is an $\Fqm$-linear AMRD code with mimimum distance $\widehat{d}=n-k+1$ and dual distance $\widehat{d}^{\perp}=1.$
\end{lemma}

\begin{proof} One easily verifies that $d(C)=d(\widehat{C})$, for all $\Fqm$-linear code $C \leq \Fqmn$. Therefore $n-k+1=d(\mathcal{G})=d(\widehat{\mathcal{G}})$. Since $\dim_{\Fqm} \widehat{\mathcal{G}}=k$, then $\rdef(\widehat{\mathcal{G}})=(n+1)-k+1-d(\widehat{\mathcal{G}})=1$. Furthermore, since $\overline{1}=(1, \ldots, 1) \in \widehat{\mathcal{G}}^\perp $, we have $\widehat{d}^{\perp}=1.$

\end{proof}
Given an $\F_q$-linear AMRD code with $m\mid t$, the following lemma allows us to find $\F_q$-linear AMRD codes for which $m$ does not divide the dimension $t$.
\begin{lemma} If $\mC \leq (\F_q)_{n,m}$ is an AMRD $\F_q$-linear code of dimension $t$ with $m, t \neq 1$ and $m \mid t$, then there exists an AMRD $\F_q$-linear code $\mC' \leq (\F_q)_{n,m}$ of dimension $t'$ with $m \nmid t'$.
\end{lemma}
\begin{proof} If $m \mid t$, then it is always possible to find an integer $t'$ such that $t'< t$ and $\lceil \frac{t'}{m} \rceil=\frac{t}{m}$. Since $\mC$ is AMRD, then $d(\mC)=n-\frac{t}{m}=n-\lceil \frac{t'}{m} \rceil$. Let $\mC'$ be the $t'$-dimensional subcode of $\mC$ containing a vector of minimum rank $d$. Then $\mC'$ is AMRD.
\end{proof}
It is well known that the dual code of an MRD code is also an MRD code and therefore both have rank defect 0. Based on this property we define.
\begin{definition} {\rm
We say that an $\F_q$-linear code $\mC$ is {dually AMRD} if
$\rdef(\mC)=\rdef(\mC^\perp)=1$. A similar definition is given for an $\Fqm$-linear code $C\leq \Fqmn$ considering its associated code.}
\end{definition}

\begin{example} \label{dually-ex} {\rm Let $q=2$, $n=m=3$, $t=4$ and $$\mC=\left \langle \left(
                                                          \begin{array}{ccc}
                                                            1 & 0 & 0 \\
                                                            0 & 0 & 0 \\
                                                            0 & 0 & 0 \\
                                                          \end{array}
                                                        \right), \left(
                                                          \begin{array}{ccc}
                                                            0 & 0 & 0 \\
                                                            0 & 1 & 0 \\
                                                            0 & 0 & 0 \\
                                                          \end{array}
                                                        \right),\left(
                                                          \begin{array}{ccc}
                                                            0 & 0 & 0 \\
                                                            0 & 0 & 0 \\
                                                            0 & 0 & 1 \\
                                                          \end{array}
                                                        \right),\left(
                                                          \begin{array}{ccc}
                                                            0 & 1 & 0 \\
                                                            0 & 0 & 0 \\
                                                            0 & 0 & 0 \\
                                                          \end{array}
                                                        \right)
 \right \rangle. $$ Then $\left(                                                        \begin{array}{ccc}
                                                            0 & 0 & 0 \\
                                                            1 & 0 & 0 \\
                                                            0 & 0 & 0 \\
                                                          \end{array}
                                                        \right) \in \mC^\perp$, $d=d^\perp=1$ and $\mC$ is dually AMRD.}

\end{example}
We can see that among all AMRD codes the dually AMRD codes are the most similar to MRD codes. However not all AMRD codes are dually AMRD, as we show in the following simple example.

\begin{example} {\rm Let $q=2$, $n=m=2$, $t=1$,  {\scriptsize$\mC=\left\{\left(
                                        \begin{array}{cc}
                                          0 & 0 \\
                                          0 & 0 \\
                                        \end{array}
                                      \right), \left(
                                        \begin{array}{cc}
                                          0 & 0 \\
                                          1 & 0 \\
                                        \end{array}
                                      \right)
 \right\}$} and {\scriptsize$$\mC^{\perp}=\left\{\left(
                                        \begin{array}{cc}
                                          0 & 0 \\
                                          0 & 0 \\
                                        \end{array}
                                      \right), \left(
                                        \begin{array}{cc}
                                          1 & 0 \\
                                          0 & 0 \\
                                        \end{array}
                                      \right), \left(
                                        \begin{array}{cc}
                                          0 & 0 \\
                                          0 & 1 \\
                                        \end{array}
                                      \right), \left(
                                        \begin{array}{cc}
                                          0 & 1 \\
                                          0 & 0 \\
                                        \end{array}
                                      \right), \left(
                                        \begin{array}{cc}
                                          1 & 0 \\
                                          0 & 1 \\
                                        \end{array}
                                      \right), \left(
                                        \begin{array}{cc}
                                          1 & 1 \\
                                          0 & 0 \\
                                        \end{array}
                                      \right), \left(
                                        \begin{array}{cc}
                                          1 & 1 \\
                                          0 & 1 \\
                                        \end{array}
                                      \right), \left(
                                        \begin{array}{cc}
                                          0 & 1 \\
                                          0 & 1 \\
                                        \end{array}
                                      \right)
 \right\}.$$} Then $d=1$ and $\mC$ is AMRD, while $d^\perp=1$ and $\mC^\perp$ is QMRD.}

\end{example}

According to Proposition 19 in \cite{DC-G-R-L} an $\Fqm$-linear code $C\leq \Fqmn$ is dually AMRD if and only if $d+d^\perp=n$. More generally we have the following results.
\begin{proposition} \label{prop 1} Let $\mC \leq (\F_q)_{n,m}$ be an $\F_q$-linear code with minimum distance $d$ and dual distance $d^\perp$.
The following facts hold:
\begin{enumerate}
\item If $\mC$ is dually AMRD, then $t\geq m$ and $d+d^\perp=\left\{
                                          \begin{array}{ll}
                                            n, & \hbox{if $m \mid t$;} \\
                                            n-1, & \hbox{if $m \nmid t$.}
                                          \end{array}
                                        \right.$

\item If $m\mid t$, then $\mC$ is dually AMRD if and only if $d+d^\perp=n$. In particular an $\Fqm$-linear code $C$ is dually AMRD if and only if $d+d^\perp=n$.
\item If $d+d^\perp=n$ and $m\nmid t$ or $d+d^\perp=n-1$ and $m\mid t$, then $\mC$ is not a dually AMRD code.
\item Let $t=\beta m+\alpha$, where $\beta, \alpha \in \Z$ and $ 0 \leq \alpha < m$. Then we have.
    \begin{enumerate}
    \item If $m \mid t$, then $\mC$ is dually AMRD if and only if $d=n-\beta$ and $d^\perp=\beta$.
    \item If $m \nmid t$, then $\mC$ is dually AMRD if and only if $d=n-\beta-1$ and $d^\perp=\beta$.

        \end{enumerate}
\item Let $t=\beta m+\alpha$, where $\beta, \alpha \in \Z$ and $ 0 \leq \alpha < m$. Then $\mC^\perp$ is AMRD if and only if $d^\perp=\beta$.
\end{enumerate}

\end{proposition}
\begin{proof}
\begin{enumerate}
\item Let $\mC$ be dually AMRD code. If $t<m$, then $\rdef(\mC^\perp)=0$, a contradiction. On the other hand, $$1=\rdef(\mC)=n-\lceil t/m \rceil
+1-d=\rdef(\mC^\perp)=n-\lceil n-t/m\rceil+1-d^\perp=\lfloor t/m \rfloor+1-d^\perp.$$ Therefore $d+d^\perp=n+\lfloor - t/m\rfloor+ \lfloor t/m\rfloor=\left\{
                                          \begin{array}{ll}
                                            n, & \hbox{if $m \mid t$;} \\
                                            n-1, & \hbox{if $m \nmid t$.}
                                          \end{array}
                                        \right.
 $

\item Let $d+d^\perp=n$. Then
$$\begin{array}{rl}
\rdef(\mC)+\rdef(\mC^\perp)= &(n-\lceil t/m\rceil+1-d)+(\lfloor t/m\rfloor+1-d^\perp)\\
=&\lfloor t/m\rfloor + \lfloor -t/m\rfloor+2\\
=&\left\{
                                                                      \begin{array}{ll}
                                                                        2, & \hbox{if $m \mid t$;} \\
                                                                        1, & \hbox{if $m \nmid t$.}
                                                                      \end{array}
                                                                    \right.
\end{array}$$
We know that if $m\mid t$ and $\rdef(\mC)=0$, then $\mC$ is MRD and $\rdef(\mC^\perp)=0$, a contradiction. Therefore $\rdef(\mC)=\rdef(\mC^\perp)=1$.
\item It is an immediate consequence of the part 1.

\item \begin{enumerate}
\item Let $m \mid t$ and $\mC$ be dually AMRD. Then $d=n- \lceil t/m\rceil=n-\beta$ and $d^\perp=\lfloor t/m\rfloor=\beta$. The reciprocal is followed from part 2.
\item Let $m \nmid t$ and $\mC$ be dually AMRD. Then $d=n- \lceil t/m\rceil=n-(\beta+1)$ and $d^\perp=\lfloor t/m\rfloor=\beta$. Reciprocally, let $d=n-\beta-1$ and $d^\perp=\beta$. Since $t/m=\beta +\alpha/m$ with $0<\alpha/m<1 $, then $\rdef(C)=n-\lceil t/m\rceil+1-d=n-(\beta+1)+1-n+\beta+1=1$ and $\rdef(C^\perp)=\lfloor t/m\rfloor+1-d^\perp=\beta+1-\beta=1.$
\end{enumerate}

\item By \cite[Lemma 21]{DC-G-R-L} we have $\rdef(\mC^\perp)= \beta +1-d^\perp$. Therefore the result easily follows. (Notice that $d^\perp=\beta$, implies $\beta \neq 0$, which is equivalent to $ t>m$).
\end{enumerate}
\end{proof}

\begin{remark}  {\rm  The reciprocal of the Proposition \ref{prop 1} (1) is not true when $m \nmid t$. In fact, if $d+d^\perp=n-1$ and $m\nmid t$, then $\rdef(C)+\rdef(C^\perp)=2$. Therefore $\left(\rdef(C), \rdef(C^\perp)\right)\in \{(0,2), (2,0), (1,1)\}$ and $\mC$ is not necessarily a dually AMRD code. For example if $q=2$, $m=n=3$, $t=7$ and $$\mC^\perp=\left\langle \left(
                           \begin{array}{ccc}
                             1 & 0 & 0 \\
                             0 & 0 & 0 \\
                             0 & 0 & 0 \\
                           \end{array}
                         \right), \left(
                           \begin{array}{ccc}
                             0 & 0 & 0 \\
                             0 & 1 & 0 \\
                             0 & 0 & 0 \\
                           \end{array}
                         \right)
 \right\rangle,$$ then $\left(
                           \begin{array}{ccc}
                             0 & 0 & 0 \\
                             1 & 0 & 0 \\
                             0 & 0 & 0 \\
                           \end{array}
                         \right) \in \mC$, $d=d^\perp=1$ and $\mC$ is QMRD. This happens because $d \neq n-\beta-1$ and $d^\perp \neq \beta.$}

\end{remark}

\begin{remark}  {\rm By Proposition \ref{prop 1} (2) an $\Fqm$-linear code $C$ and its associated code $\lambda_{\mathcal{B}}(C)$ are dually AMRD if and only if $d+d^\perp=n$. However not always a dually AMRD $\F_q$-linear code $\mC$ arises from a dually AMRD $\Fqm$-linear code $C$, even when $m \mid t$ and $d+d^\perp=n$. For example the code $\mC \leq (\F_2)_{3,3}$ with dimension $t=3$, $d+d^\perp=2+1=3$ and whose nonzero codewords are {\scriptsize $$ \left(
                           \begin{array}{ccc}
                             1 & 0 & 0 \\
                             0 & 1 & 0 \\
                             0 & 0 & 0 \\
                           \end{array}
                         \right), \left(
                           \begin{array}{ccc}
                             0 & 1 & 0 \\
                             0 & 0 & 1 \\
                             1 & 0 & 0 \\
                           \end{array}
                         \right), \left(
                           \begin{array}{ccc}
                             0 & 0 & 1 \\
                             1 & 0 & 0 \\
                             0 & 1 & 0 \\
                           \end{array}
                         \right), \left(
                           \begin{array}{ccc}
                             1 & 1 & 0 \\
                             0 & 1 & 1 \\
                             1 & 0 & 0 \\
                           \end{array}
                         \right), \left(
                           \begin{array}{ccc}
                             1 & 0 & 1 \\
                             1 & 1 & 0 \\
                             0 & 1 & 0 \\
                           \end{array}
                         \right), \left(
                           \begin{array}{ccc}
                             0 & 1 & 1 \\
                             1 & 0 & 1 \\
                             1 & 1 & 0 \\
                           \end{array}
                         \right),\left(
                           \begin{array}{ccc}
                             1 & 1 & 1 \\
                             1 & 1 & 1 \\
                             1 & 1 & 0 \\
                           \end{array}
                         \right),
$$} is a dually AMRD $\F_2$-linear code which does not arise from an $\F_{2^3}$-linear code $C$. In fact, suppose $\lambda_{\mathcal{B}}(C)=\mC$ for some code $C \leq \F_{2^3}^3$ and $\mathcal{B}=\{\gamma_1, \gamma_2, \gamma_3\}$ is a basis of $\F_{2^3}$ over $\F_2$. Then there exists a codeword $c\in C$ such that $c':=\left(
                           \begin{array}{ccc}
                             1 & 0 & 0 \\
                             0 & 1 & 0 \\
                             0 & 0 & 0 \\
                           \end{array}
                         \right)=\lambda_{\mathcal{B}}(c)$. Therefore $c=(\gamma_1, \gamma_2,0)\in C$ and $0 \neq \gamma_1^{-1}c=(1, \gamma_1^{-1}\gamma_2, 0)\in C $. Then $0, c' \neq \lambda_{\mathcal{B}}(\gamma_1^{-1}c)=\left(\begin{array}{ccc}
                             \lambda_{11} & \lambda_{12} & \lambda_{13} \\
                             \lambda_{21} & \lambda_{22} & \lambda_{23} \\
                             0 & 0 & 0 \\
                           \end{array}
                         \right)\in \mC$, which is not possible.  }

\end{remark}
\begin{theorem}[Existence of 1-dimensional dually AMRD codes]\label{exist-dually-1} Dually AMRD $\Fqm$-linear codes $C\leq \Fqmn$ with dimension $k=1$ exist for all parameters $m,n,q$.

\end{theorem}
\begin{proof} Let $v_1, \ldots, v_{n-1}$ be linearly independent and $\mathcal{G}$ the Gabidulin code generated by $M_1(v_1, \ldots, v_{n-1})$. Then $\widehat{\mathcal{G}}$ is a $1$-dimensional dually AMRD code. In fact, $d(\widehat{\mathcal{G}})=d(\mathcal{G})=n-1$ and $d(\widehat{\mathcal{G}}^\perp)=1$. Therefore $d(\widehat{\mathcal{G}})+d(\widehat{\mathcal{G}}^\perp)=n$ and by Proposition \ref{prop 1} $\widehat{\mathcal{G}}$ is dually AMRD.
  \end{proof}

\section{Rank distribution of dually AMRD $\F_q$-linear codes}
In the following theorem the authors proved in \cite{DC-G-R-L} that the rank distribution of a code $\mC$ is determined by
its parameters, together with the number of codewords of small weight:
$A_d(\mC),\ldots, A_{n-d^\bot}(\mC)$.
\begin{theorem}[\cite{DC-G-R-L}]\label{04-02-15}
Let $\mC \leq (\F_q)_{n,m}$ be a $t$-dimensional code, with minimum distance $d$ and dual minimum distance $d^\perp$.
Let $\delta=1$ if $\mC$ is MRD, and $\delta=0$ otherwise.
For all $1 \le r \le d^\perp$ we have
\begin{eqnarray*}
A_{n-d^\bot+r}(\mC) &=& (-1)^{r}q^{r \choose 2}
\sum\limits_{j=d^\bot}^{n-d}{j\brack d^\bot-r}  {j-d^\bot+r-1\brack
r-1}A_{n-j}(\mC)\\
&& + {n\brack d^\bot-r}\sum\limits_{i=0}^{r-1-\delta} (-1)^iq^{i \choose 2
}{n-d^\bot+r\brack i}    \left(q^{t-m(d^\perp-r+i)}-1\right).
\end{eqnarray*}
In particular, $n$, $m$, $t$, $d$, $d^\perp$ and $A_d(\mC),\ldots, A_{n-d^\bot}(\mC)$
determine the rank
distribution of $\mC$.
\end{theorem}
If we apply this theorem to dually AMRD $\F_q$-linear codes, then we have the following results.
\begin{proposition}\label{prop 0} Let $\mC \leq (\F_q)_{n,m}$ be a $t$-dimensional dually AMRD code, with minimum distance $d$ and dual minimum distance $d^\perp$.
The following facts hold:

\begin{enumerate}
\item If $m \nmid t$ and $t=\beta m + \alpha$ with $\beta, \alpha \in \Z$ and $1 \leq \alpha < m$, then \begin{eqnarray*}
A_{d+r}(\mC) &=& (-1)^{r-1}q^{r -1\choose 2} \left(
{n-d-1\brack r-1}  A_{d+1} + {n-d\brack r} {r-1\brack 1}  A_{d} \right)\\
 & &+{n\brack d+r}\sum\limits_{i=0}^{r-2} (-1)^iq^{i \choose 2
}{d+r\brack i}    \left(q^{\alpha+m(r-i)}-1\right),
\end{eqnarray*}
for all $r=2,\ldots,n-d$.
\item If $m \mid t$, then we have
\begin{eqnarray*}
A_{d+r}(\mC) &=& (-1)^{r}q^{r \choose 2}
{n-d\brack r}  A_d+ {n\brack d+r}\sum\limits_{i=0}^{r-1} (-1)^iq^{i \choose 2
}{d+r\brack i}    \left(q^{m(r-i)}-1\right),
\end{eqnarray*}
for all $r=1,\ldots,n-d$.
In particular, for an $\Fqm$-linear code $C\leq \Fqmn$ of dimension $k$ we have
\begin{eqnarray*}
A_{d+r}(\mC) &=& (-1)^{r}q^{r \choose 2}
{k\brack r}  A_d+ {n\brack k-r}\sum\limits_{i=0}^{r-1} (-1)^iq^{i \choose 2
}{d+r\brack i}    \left(q^{m(r-i)}-1\right),
\end{eqnarray*}
for all $r=1,\ldots,n-d$.
\item If $m \mid t$, then $\frac{{n\brack d+2}}{q{n-2\brack 2}} (q^m-1)\left({d+2\brack 1}-q^m-1\right) \leq A_d\leq \frac{{n\brack d+1}}{{n-d\brack 1}}(q^{m}-1)$.
\item Let $m \mid t$. If $A_{d+1}=0$, then $A_d=\frac{{n\brack d+1}}{{n-d\brack 1}}(q^{m}-1)$ and if $A_{d+2}=0$, then $$ A_d= \frac{{n\brack d+2}}{q{n-2\brack 2}}(q^m-1) \left({d+2\brack 1}-q^m-1\right).$$
\end{enumerate}
\end{proposition}
\begin{proof}
The proof of parts 1 and 2 are immediate by Theorem \ref{04-02-15}. The proof of parts 3 and 4 follow from
$$ A_{d+1}=-{n-d \brack 1}A_d+ {n \brack d+1}(q^{m}-1) \geq 0$$ and $$ A_{d+2}=q{n-d \brack 2}A_d+ {n \brack d+2}(q^{m}-1)\left ( q^m+1- {d+2 \brack 1}\right ) \geq 0.$$
\end{proof}
In Theorem 35 \cite{DC-G-R-L} it was proved that if $\mC$ is dually AMRD
and $m \mid t$, then the number of codewords of minimum rank in $\mC$ and $\mC^\perp$ are equal. We present a more general result through a different proof.

\begin{lemma} \label{lema 2} Let $\mC \leq (\F_q)_{n,m}$ be a $t$-dimensional dually AMRD code, with minimum distance $d$ and dual minimum distance $d^\perp$. Then
$$A_d=q^{t-m \lceil t/m\rceil} \left ( {n\brack \lceil t/m \rceil}+ {n- \lceil t/m\rceil+1\brack 1} A^\perp_{\lceil t/m\rceil-1} + A^\perp_{\lceil t /m \rceil} \right )- {n \brack \lceil t/m\rceil}.$$
More precisely, we have:
\begin{enumerate}
\item $A_d=(q^{\alpha-m}-1){n \brack \beta+1}+q^{\alpha-m}\left( { n-\beta \brack 1}A_{\beta}^\perp+A_{\beta+1}^\perp\right)$ with $d^\perp=\beta$, if $m \nmid t$ and $t=m\beta+\alpha > m$, where $\beta, \alpha \in \Z$ and $ 0 \leq \alpha < m$.
\item $A_d=A^\perp_{d^\perp}$, if  $m\mid t$.
\end{enumerate}

%
%
\end{lemma}

\begin{proof}
 We know that $d=n - \lceil t/m \rceil$ and $d^\perp= \lfloor t/m \rfloor$. By Theorem 31 \cite{Alberto} we have $$ \sum_{i=0}^{n-\nu}{n-i \brack \nu } A_i= q^{t-m\nu} \sum_{j=0}^{\nu} {n-j \brack \nu-j} A^\perp_{j},$$ for all $ 0 \leq \nu  \leq n$. In particular, for $\nu=\lceil t/m \rceil$ we have
$$ \sum_{i=0}^{n- \lceil t/m \rceil}{n-i \brack \lceil t/m \rceil } A_i= q^{t-m\lceil t/m \rceil)} \sum_{j=0}^{\lceil t/m \rceil} {n-j \brack \lceil t/m \rceil-j} A^\perp_{j}.$$ Therefore

$$ {n \brack \lceil t/m\rceil}+A_d=q^{t-m \lceil t/m\rceil} \left ( {n\brack \lceil t/m \rceil}+ {n- \lceil t/m\rceil+1\brack 1} A^\perp_{\lceil t/m \rceil-1} + A^\perp_{\lceil t /m \rceil} \right ) .$$

\end{proof}
\begin{theorem}\label{teo-coef} Let $\mC \leq (\F_q)_{n,m}$ be a $t$-dimensional AMRD code with $t=m\beta+\alpha > m$, where $\beta, \alpha \in \Z$ and $ 0 \leq \alpha < m$. The following hold:
\begin{enumerate}
\item If $ m\nmid t$, then $\mC$ is dually AMRD if and only if $$A_d=(q^{\alpha-m}-1){n \brack \beta+1}+q^{\alpha-m}\left( { n-\beta \brack 1}A_{\beta}^\perp+A_{\beta+1}^\perp\right)\;\; with\;\; A_{\beta}^\perp \neq 0.$$
\item If $m \mid t$, then $\mC$ is dually AMRD if and only if $A_d=A^\perp_{d^\perp}.$
\end{enumerate}

\end{theorem}
\begin{proof} \begin{enumerate}
\item Let $A_d=(q^{\alpha-m}-1){n \brack \beta+1}+q^{\alpha-m}\left( { n-\beta \brack 1}A_{\beta}^\perp+A_{\beta+1}^\perp\right)$ with $A_{\beta}^\perp \neq 0$. By the Singleton bound $d^\perp \leqslant \beta+1$. Therefore, $d^\perp \leqslant \beta$. Since $\mC$ is AMRD, then $d=n-(\beta+1)$ and we have $$\sum_{i=0}^{n-(\beta+1)}{n-i \brack \beta+1 } A_i= q^{\alpha-m} \sum_{j=0}^{\beta+1} {n-j \brack \beta+1-j} A^\perp_{j}.$$ It follows $A_d=(q^{\alpha-m}-1){n \brack \beta+1}+q^{\alpha-m} \sum_{j=d^\perp}^\beta { n-j \brack \beta+1-j}A_{j}^\perp+q^{\alpha-m}A_{\beta+1}^\perp.$ Suppose $d^\perp< \beta $, then
$$A_d=(q^{\alpha-m}-1){n \brack \beta+1}+q^{\alpha-m} \sum_{j=d^\perp}^{\beta-1} { n-j \brack \beta+1-j}A_{j}^\perp+q^{\alpha-m} {n-\beta \brack 1}A_{\beta} ^\perp+q^{\alpha-m}A_{\beta+1}^\perp.$$ Then we have a contradiction since $\sum_{j=d^\perp}^{\beta-1} { n-j \brack \beta+1-j}A_{j}^\perp>0.$ Hence $d^\perp=\beta $ and $\mC^\perp$ is AMRD.

\item Let $A_d=A^\perp_{d^\perp}$ and $\delta=t/m \in \Z$. Since $\mC$ is AMRD, then $d=n-\delta$. Therefore $$\sum_{i=0}^{n- \delta}{n-i \brack \delta } A_i=  \sum_{j=0}^{\delta} {n-j \brack \delta-j} A^\perp_{j}.$$ Then we have $A_d=\sum_{j=1}^{\delta} {n-j \brack \delta-j} A^\perp_{j}.$  By the Singleton bound $d^\perp\leq \delta+1$. If $d^\perp=\delta+1$, then $A_d=0$, a contradiction. On the other hand, if $d^\perp <\delta$, then $A_d= {n-d^\perp \brack \delta-d^\perp}A_d+\sum_{j=d^\perp+1}^{\delta} {n-j \brack \delta-j} A^\perp_{j},$ a contradiction again. Hence $d^\perp=\delta$ and by Proposition \ref{prop 1} (5) we have $\mC^\perp$ is AMRD.
\end{enumerate}
\end{proof}
\begin{example}{\rm

\begin{enumerate}
\item In the Example \ref{dually-ex} we have $d=1$, $d^\perp=1$, $A_1=6$, $A_2=7$, $A_3=2$, $A^{\perp}_1=9$, $A^{\perp}_2=18$, $A_3^{\perp}=4$
and $A_d=(\frac{1}{4}-1) {3 \brack 2}+\frac{1}{4}\left( {2 \brack 1}9+18\right)$, which says that $\mC$ is dually AMRD.
 \item If $q=3$, $m=n=3$, $t=3$ and $$\mC=\left\langle \left(
                           \begin{array}{ccc}
                             1 & 0 & 0 \\
                             0 & 1 & 0 \\
                             0 & 0 & 0 \\
                           \end{array}
                         \right), \left(
                           \begin{array}{ccc}
                             0 & 1 & 0 \\
                             0 & 0 & 1 \\
                             1 & 0 & 0 \\
                           \end{array}
                         \right), \left(
                           \begin{array}{ccc}
                             0 & 0 & 1 \\
                             1 & 0 & 0 \\
                             0 & 1 & 0 \\
                           \end{array}
                         \right)
 \right\rangle,$$ then $\left(
                           \begin{array}{ccc}
                             0 & 0 & 0 \\
                             0 & 0 & 0 \\
                             0 & 0 & 2 \\
                           \end{array}
                         \right) \in \mC^\perp$, $d= 2$, $d^\perp=1$, $\mC$ is dually AMRD and $A_2=A_1^\perp=6$.
\end{enumerate} }

\end{example}

\begin{remark} {\rm It is possible that $\mC$ is a $t$-dimensional AMRD code, $t\geq m$, $m \nmid t$ and $A_d=(q^{\alpha-m}-1){n \brack \beta+1}+q^{\alpha-m}
\left( { n-\beta \brack 1}A_{\beta}^\perp+A_{\beta+1}^\perp\right)$, but $\mC$ is not dually AMRD. For example, if $q=2$, $m=n=3$, $t=5$ and $$\mC=\left\langle \left(
                           \begin{array}{ccc}
                             1 & 0 & 0 \\
                             0 & 0 & 0 \\
                             0 & 0 & 1 \\
                           \end{array}
                         \right), \left(
                           \begin{array}{ccc}
                             0 & 1 & 0 \\
                             0 & 1 & 1 \\
                             0 & 1 & 0 \\
                           \end{array}
                         \right), \left(
                           \begin{array}{ccc}
                             0 & 0 & 1 \\
                             0 & 1 & 0 \\
                             0 & 0 & 0 \\
                           \end{array}
                         \right),\left(
                           \begin{array}{ccc}
                             0 & 0 & 0 \\
                             1 & 0 & 0 \\
                             0 & 1 & 0 \\
                           \end{array}
                         \right),\left(
                           \begin{array}{ccc}
                             0 & 0 & 0 \\
                             0 & 0 & 0 \\
                             1 & 0 & 0 \\
                           \end{array}
                         \right)
 \right\rangle,$$ then $d=1$, $d^\perp=2$, $A_1=1$, $A_2=18$, $A_3=12$, $A^{\perp}_2=9$, $A_3^{\perp}=6$. Therefore $\mC$ is AMRD but it is not dually AMRD, even though $A_d=A_1=(2^{-1}-1){ 3 \brack 2}+2^{-1} \left( 9 \right)$. Note that this does not contradict the Theorem \ref{teo-coef} because in this case $A^\perp_{\beta}(\mC^\perp)=0$.}
\end{remark}
\begin{lemma} Let $\mC \leq (\F_q)_{n,m}$ be a self-dual AMRD code. The following hold:
\begin{enumerate}
\item If $2 \nmid n$, then $d=\frac{n-1}{2}$, $m \equiv 0 \mod 2$ and $A_{d+1}=(q^{m/2}-{\frac{n+1}{2} \brack 1})A_d-(1-q^{m/2}) {n \brack \frac{n+1}{2}}$.

\item If $2 \mid n$, then $d=n/2.$
\end{enumerate}

\end{lemma}

\begin{proof} We know that $2 \mid n$ if and only if $m \mid t=\frac{nm}{2}.$ Therefore, if $2 \nmid n$, then $d=n- \lceil t/m\rceil=\frac{n-1}{2}$. The rest of the statement in part 1 follows from Lemma \ref{lema 2}. Similarly, if $2 \mid n$, then $d=n- \lceil t/m\rceil=\frac{n}{2}.$

\end{proof}

\begin{corollary} Let $\mC \leq (\F_q)_{n,m}$ be a self-dual AMRD code and $\chara(\F_q)=2$. Then $\mC$ has parameters
\begin{enumerate}
 \item $n=3$, $t=3\frac{m}{2}$, $d=1$ and $4 \leq m \equiv 0 \mod 2$ or
 \item $n=2$, $t=m$ and $d=1$.
\end{enumerate}
\end{corollary}
\begin{proof} By \cite[Theorem 1]{Nebe-Willems} the all-ones matrix is in a self-dual code $\mC$. Therefore $d(\mC)=1$.
\end{proof}
\begin{example} Let $q=2$, $n=2$, $m=3$ and $$\mC=\left\langle \left(
                           \begin{array}{ccc}
                             1 & 1 & 0 \\
                             0 & 0 & 0 \\
                           \end{array}
                         \right), \left(
                           \begin{array}{ccc}
                             0 & 0 & 0 \\
                             1 & 1 & 0 \\
                           \end{array}
                         \right), \left(
                           \begin{array}{ccc}
                             0 & 0 & 1 \\
                             0 & 0 & 1 \\
                           \end{array}
                         \right)
 \right\rangle.$$ Then $\mC$ is a $3$-dimensional self-dual AMRD code.

\end{example}
 Since in general $\lambda _{\mathcal{B}}(C^\perp) \neq \lambda _{\mathcal{B}}(C)^\perp$, if $C$ is a self-dual AMRD code, the associated code $\lambda _{\mathcal{B}}(C)\leq (\F_q)_{n,m}$ is not necessarily a self-dual AMRD code. However if $\chara(\F_q)=2$, the statement is true.

\begin{theorem} Let $C \leq \Fqmn$ be a $k$-dimensional self-dual AMRD code. The following hold:
\begin{enumerate}
\item $\dim_{\F_q}(\lambda _{\mathcal{B}}(C))= \dim_{\F_q}(\lambda _{\mathcal{B}}(C))^\perp.$
\item $C$ is a $[2d, d,d]$ code.
\item If $n=2$, then $C$ is a $[2,1,1]$ code with $C=\langle(\alpha, 1)\rangle$, where $\alpha \in \Fq$ and $\alpha^2=-1$.
\item If $\chara(\F_q)=2$, then $C=\langle(1, 1)\rangle$ and $\lambda _{\mathcal{B}}(C)\leq (\F_q)_{2,m}$ is a self-dual code. 
\end{enumerate}

\end{theorem}

\begin{proof}
\begin{enumerate}
\item $\dim_{\F_q}(\lambda _{\mathcal{B}}(C))=m\cdot k=\frac{nm}{2}$ and $\dim_{\F_q}(\lambda _{\mathcal{B}}(C)^\perp)=nm-\dim_{\F_q}(\lambda _{\mathcal{B}}(C))=nm- mk=nm-\frac{nm}{2}=\frac{nm}{2}$.
\item Since $d=n-k=n/2$ and $k=d$, the result easily follows.
\item By part 2 we have that $C$ is a $[2,1,1]$ code. Let $C=\langle(x, y)\rangle$, where $x,y\in \Fqm$. Then $x^2+y^2=0$ and $\dim_{\F_q}\langle x,y\rangle=1$. Therefore there exists $\alpha \in \F_q$ such that $x=\alpha y$ and we have $(x,y)=y(\alpha,1)$ with $1+\alpha^2=0.$
\item If $\chara(\F_q)=2$, then $\overline{1}:=(1, \ldots, 1)\in C$. In fact, if $c=(c_1, \ldots, c_n)\in C$, then $0=c \cdot c=\sum_{i=1}^n c_i^2=\left(\sum_{i=1}^n c_i \right)^2$. Therefore $0=\sum_{i=1}^n c_i=\overline{1} \cdot c$. Hence we have $d(C)=1$ and $C=\langle(1, 1)\rangle$ is a $[2,1,1]$ code. On the other hand, if $\lambda _{\mathcal{B}}(c)$, $\lambda _{\mathcal{B}}(c') \in \lambda _{\mathcal{B}}(C)$, then $\lambda _{\mathcal{B}}(c)=\left(
                                                           \begin{array}{ccc}
                                                             \lambda_{11} & \ldots & \lambda_{1m}\\
                                                             \lambda_{11} & \ldots & \lambda_{1m}\\
                                                           \end{array}
                                                         \right)
    \in (\F_q)_{2,m}$ and $\lambda _{\mathcal{B}}(c')=\left(
                                                           \begin{array}{ccc}
                                                             \lambda'_{11} & \ldots & \lambda'_{1m}\\
                                                             \lambda'_{11} & \ldots & \lambda'_{1m}\\
                                                           \end{array}
                                                         \right)
    \in (\F_q)_{2,m}$. Therefore, $\textrm{Tr}\left(\lambda _{\mathcal{B}}(c) \lambda _{\mathcal{B}}(c')^T\right)=0$ and $\lambda _{\mathcal{B}}(C) \subseteq \lambda _{\mathcal{B}}(C)^\perp$. Since $\dim_{\F_q}(\lambda _{\mathcal{B}}(C))= \dim_{\F_q}(\lambda _{\mathcal{B}}(C)^\perp)$, then $$\lambda _{\mathcal{B}}(C)= \lambda _{\mathcal{B}}(C)^\perp.$$     
\end{enumerate}
\end{proof}

\begin{lemma} \label{lema 3} Let $\mC \leq (\F_q)_{n,m}$ be a $t$-dimensional dually AMRD code, with minimum distance $d$ and $t=\frac{nm}{2}=dm$. Then $\mC$ is a formally self-dual code. In particular, if $C \leq \Fqmn$ is a dually AMRD $\Fqm$-linear $[2d,d, d]$ code with $n=2d\leq m$, then $C$ is formally self-dual.
\end{lemma}
\begin{proof}
 Since $d=n-t/m$ and $d^\perp=t/m$, then $d=d^\perp$. By Theorem \ref{teo-coef} (2) we have $A_d(\mC)=A_{d^\perp}(\mC^\perp)=A_{d}(\mC^\perp)$. Therefore by Proposition \ref{prop 0} (2) we have $A_i(\mC)=A^\perp_{i}(\mC^\perp)$ for $i=0,\ldots, n$.
\end{proof}

\begin{lemma} \label{lema 4}
Let $\mC \leq (\F_q)_{n,m}$ be a $t$-dimensional dually AMRD code, with minimum distance $d$ and dual minimum distance $d^\perp$. If $m \mid t$ and $A_{d+1}=0$, then $t \leq \frac{nm}{2}$. In particular if $C \leq \Fqmn$ is a  dually AMRD $\Fqm$-linear $[n,k, d]$ code and $A_{d+1}=0$, then $k \leq n/2$.
\end{lemma}

\begin{proof} Let $\beta:=t/m \in \Z$. By Proposition \ref{prop 0} (2) we have $A_d\leq \frac{{n\brack d+1}}{{n-d\brack 1}}(q^{m}-1)$, which is equivalent to $A_d \leq \frac{{n\brack \beta -1}}{{\beta \brack 1}}(q^m-1)$. Similarly $A_{d^\perp} \leq \frac{{n\brack \beta +1}}{{n-\beta \brack 1}}(q^m-1)$. Then by Proposition \ref{prop 0} (3) and Theorem \ref{teo-coef} (2) we have $\frac{{ n \brack \beta +1}}{ { n-\beta \brack 1}} \leq \frac{{ n \brack \beta -1}}{ { \beta \brack 1}} $. Therefore $n-\beta \geq \beta$.

\end{proof}
\section{Generalized weights of dually AMRD codes}
In $\cite{Ducoat2}$ the authors define an $i$-MRD $\Fqm$-linear code $C \leq \Fqmn$ as a code $C$ meeting the generalized Singleton bound i.e. $\mathcal{M}(C)=n-k+i$. Additionally we say that an $i$-MRD $\Fqm$-linear code $C \leq \Fqmn$ is of \textit{degree} $\deg(C)=i-1$ if $i$ is the minimum integer with this property. For $\F_q$-linear codes we define:
\begin{definition}\label{r-MRD} {\rm Let $\mC \leq (\F_q)_{n,m}$ be an $\F_q$-linear code of dimension $t$. If $a_r(\mC)=n-\lfloor \frac{t-r}{m}\rfloor$ for an integer $r=1+(i-1)m \in \{1,1+m, 1+2m, \ldots, 1+(\lceil t/m\rceil-1)m\}$, we say that $\mC$ is an \textit{$i$-MRD $\F_q$-linear code with $r=1+(i-1)m$}. If $i$ is the minimum integer with this property, $\mC$ is called an $i$-MRD $\F_q$-linear code of \textit{degree} $\deg(\mC)=i-1$.}
\end{definition}
One easily verifies from the definition that MRD and QMRD $\F_q$-linear codes are $1$-MRD codes of degree 0.
In the case of $\Fqm$-linear codes, if $C \leq \Fqmn$ is an $i$-MRD $\Fqm$-linear code of degree $\deg(C)=i-1$ and the dimension of its associated code $\lambda_{\mathcal{B}}(C)$ is $t$, then by Theorem \ref{theorem 30} and Theorem \ref{corollary 31}
we have  $$n- \left\lfloor \frac{t-(1+(i-1)m)}{m}\right\rfloor= n-k+i=\mathcal{M}_i(C)=a_{im}(\lambda_{\mathcal{B}}(C))=a_{1+(i-1)m}(\lambda_{\mathcal{B}}(C)), $$
i.e. $\lambda_{\mathcal{B}}(C)$ is an $i$-MRD $\F_q$-linear code of degree $i-1$. The reciprocal is also true and we have that $C$ is an $i$-MRD $\Fqm$-linear code of degree $i-1$ if and only if $\lambda_{\mathcal{B}}(C)$ is an $i$-MRD $\F_q$-linear code of degree $i-1$. Therefore Definition \ref{r-MRD} is an appropriate generalization of the concept of $i$-AMR for $\Fqm$-linear codes.

We know that if an $\Fqm$-linear code $C$ is $i$-MRD, then $C$ is an $(i+1)$-MRD $\Fqm$-linear code. For $\F_q$-linear codes we have the following result.
\begin{lemma}\label{domino} If $\mC \leq (\F_q)_{n,m}$ is a $t$-dimensional $i$-MRD $\Fq$-linear code with $r=1+(i-1)m$, then $\mC$ is an $(i+1)$-MRD code for $1 \leq i \leq \lceil t/m\rceil-1$.
\end{lemma}

\begin{proof}
We prove the statement by induction over $i$. One easily verifies it for $i=1$. Assume $a_r=n-\left\lfloor \frac{t-(1+(i-1)m)}{m}\right\rfloor$ for $r=1+(i-1)m.$. Then we have by Theorem \ref{theorem 30}
 $$ a_r=n-\left\lfloor\frac{t-(1+(i-1)m)}{m}\right\rfloor < a_{r+m}\leq n- \left\lfloor\frac{t-(1+im)}{m}\right\rfloor.$$ Then $$ n-\left\lfloor\frac{t-(1+(i-1)m)}{m}\right\rfloor < a_{r+m}\leq n- \left\lfloor \frac{t-(1+(i-1)m)}{m}\right\rfloor+1.$$ Therefore $a_{1+im}=a_{r+m}= n- \left\lfloor \frac{t-(1+(i-1)m)}{m}\right\rfloor+1= n- \left\lfloor\frac{t-(1+im)}{m}\right\rfloor.$
\end{proof}

\begin{lemma}\label{lemma-mk=n} Let $\mC \leq (\F_q)_{n,m}$ be an $\Fq$-linear code of dimension $t$. The following facts are equivalent:
\begin{enumerate}
\item $a_{1+(\lceil t/m\rceil-1)m}(\mC)=n$.
\item $a_{t+1-\lfloor t/m\rfloor m}(\mC^\perp) \neq 1$.
\item There exists $i \in \{1,\ldots, \lceil t/m\rceil\}$ such that $\mC$ is an $i$-MRD code of degree $0 \leq \deg(\mC) \leq i-1 \leq \lceil t/m\rceil-1$.
\end{enumerate}

\end{lemma}

\begin{proof} Interchanging $\mC$ with $\mC^\perp$ in Theorem \ref{theorem 32} we have
 $$\{n+1-a^{\perp}_{1+t+(n-\lceil \frac{2t+1}{m}\rceil)m}, \ldots, n+1-a^{\perp}_{1+t-\lfloor t/m\rfloor m}\}=[n] \backslash\{ a_1, a_{1+m}, \ldots, a_{1+(\lceil t/m\rceil-1)m}\}.$$
 Since $$n+1-a^{\perp}_{1+t+(n-\lceil \frac{2t+1}{m}\rceil)m}< \ldots< n+1-a^{\perp}_{1+t-\lfloor t/m\rfloor m},$$ then
 $$\max\{n+1-a^{\perp}_{1+t+(n-\lceil \frac{2t+1}{m}\rceil)m}, \ldots, n+1-a^{\perp}_{1+t-\lfloor t/m\rfloor m}\}=n+1-a^{\perp}_{1+t-\lfloor t/m\rfloor m}.$$
 Similarly $\max\{ a_1, a_{1+m}, \ldots, a_{1+(\lceil t/m\rceil-1)m}\}=a_{1+(\lceil t/m\rceil-1)m}$.\\
  Therefore $a_{1+(\lceil t/m\rceil-1)m}=n \Leftrightarrow n+1-a^{\perp}_{1+t-\lfloor t/m\rfloor m} \neq n \Leftrightarrow a^{\perp}_{1+t-\lfloor t/m\rfloor m} \neq 1$.

On the other hand, if $a_{1+(\lceil t/m\rceil-1)m}=n$, then $a_{1+(\lceil t/m\rceil-1)m}= n - \left \lfloor \frac{t-(1+(\lceil t/m\rceil-1)m)}{m} \right \rfloor$, since $$ n - \left \lfloor \frac{t-(1+(\lceil t/m\rceil-1)m)}{m} \right \rfloor =n - \left\lfloor \frac{t-1}{m}\right\rfloor+ \lceil t/m\rceil-1 = n -\lceil t/m\rceil+1+ \lceil t/m\rceil-1=n.$$
Thus $\mC$ is $\lceil t/m\rceil$-MRD. Reciprocally, if $\mC$ is $i$-MRD, for $i \in \{1, \ldots, \lceil t/m\rceil-1\}$, then $\mC$ is $(i+j)$-MRD
for $j \geq 1$. In particular $\mC$ is $\lceil t/m\rceil$-MRD code and $a_{1+(\lceil t/m\rceil-1)m}=n.$
\end{proof}

\begin{corollary}\label{corollary-mk=n} Let $C\leq \Fqmn$ be an $\Fqm$-linear code. The following facts are equivalent:
\begin{enumerate}
\item $\m_k(C)=n$.
\item $\m_1(C^\perp) \neq 1$.
\item There exists $i\in\{1,\ldots, k\}$, such that $C$ is an $i$-MRD code of degree $0\leq \deg(C) \leq i-1 \leq k-1$.
\end{enumerate}

\end{corollary}
\begin{proof} In this case $t:=\dim_{\F_q}(\lambda_{\mathcal{B}}(C))=m\cdot k$. Then by Theorem \ref{corollary 31} we have  $$\m_k(C)=a_{mk}(\lambda_{\mathcal{B}}(C))=a_t(\lambda_{\mathcal{B}}(C))=a_{1+(t/m-1)m}
(\lambda_{\mathcal{B}}(C)).$$

\end{proof}
\begin{remark} {\rm \begin{enumerate}  \item Note that in general if $\mC$ is an $i$-MRD $\Fq$-linear code and $m \mid t$, by Lemma \ref{lemma-mk=n} this is equivalent to $a_1(\mC^\perp)\neq 1$, but not necessarily $a_{t/m}(\mC)=n.$ \item When we work with the Hamming metric we have $d(C^\perp)=1$ if and only if all vectors of $C$ have a zero at a certain position, which is equivalent to $d_k(C)=0.$ In rank metric codes, if $C\leq \Fqmn$ is an $\Fqm$-linear code, then $\m_1(C^\perp)=1$ if and only if there exists an $i$-th position, such that for every vector $v=(v_1, \ldots, v_n)\in C$ we have $v_i=\sum_{i \neq j}^n \beta_j v_j$, where $\beta_j\in \F_q.$
\end{enumerate}
}
\end{remark}


\begin{theorem}\label{Mariangel} Let $\mC\leq (\F_q)_{n,m} $ be an $\F_q$-linear code and  $a_{1+(\lceil t /m\rceil)-1)m}(\mC)=n$. Then $\mC$ is $i$-MRD with $r=1+(i-1)m$ if and only if $ 2+ \lfloor \frac{t-r}{m}\rfloor \leq a^{\perp}_{t+1-\lfloor t/m\rfloor m} $. In particular, $\mC$ is $i$-MRD of degree $i-1$ if and only if $a_{t+1-\lfloor t/m\rfloor m}(\mC^{\perp})=2+ \lfloor \frac{t-r}{m}\rfloor$ or equivalently $a_{t+1-\lfloor t/m\rfloor m}(\mC^{\perp})=\lceil t/m\rceil-i+2$.
\end{theorem}

\begin{proof}
Let $a_{1+(\lceil t/m\rceil-1)m}(\mC)=n$. We know that
 $$\{n+1-a^{\perp}_{1+t+(n-\lceil \frac{2t+1}{m}\rceil)m}, \ldots, n+1-a^{\perp}_{1+t-\lfloor t/m\rfloor m}\}=[n] \backslash
\{ a_1, a_{1+m}, \ldots, a_{1+(\lceil t/m\rceil-1)m}\}.$$ Then $n+1-a_{1+t-\lfloor t/m\rfloor m}(\mC^\perp)=\max \left( [n]
\backslash \left\{ a_1, a_{1+m}, \ldots, a_{1+(\lceil t/m\rceil-1)m}\right\}\right)$. If the
sequence \begin{equation}\label{eq 1} a_r(\mC) < a_{r+m}(\mC) < a_{r+2m}(\mC)<\ldots < a_{1+(\lceil t/m\rceil-1)m}(\mC)=n \end{equation} with $r=1+(i-1)m$ has no gaps, then $a_r(\mC)=n-\left((\lceil t/m\rceil-1)-(i-1)\right)$ and $a_r(\mC)= n-\lfloor \frac{t-r}{m}\rfloor$. Therefore by Lemma \ref{domino} the sequence (\ref{eq 1}) has no gaps if and only if $a_r(\mC)= n-\lfloor \frac{t-r}{m}\rfloor.$
Hence $n+1-a_{1+t-\lfloor t/m\rfloor m}(\mC^\perp) \leq  n-\lfloor \frac{t-r}{m}\rfloor-1$ if and only if $a_r=n-\lfloor \frac{t-r}{m}\rfloor$.

In particular, $$\begin{array}{rl}
2+ \left\lfloor \frac{t-r}{m}\right\rfloor  = a_{1+t-\lfloor t/m\rfloor m}(\mC^\perp)\Leftrightarrow & n+1-a_{1+t-\lfloor t/m\rfloor}(\mC^\perp)=n-\lfloor \frac{t-r}{m}\rfloor-1\\
\Leftrightarrow&  a_{r-m}(\mC)<a_r(\mC)-1=n-\left\lfloor \frac{t-r}{m}\right\rfloor-1=n- \left \lfloor\frac{t-(r-m)}{m} \right \rfloor,
\end{array}$$
which means that $\mC$ is $i$-MRD of degree $i-1$.
\end{proof}

\begin{corollary}\label{coro-Mariangel} Let $C\leq \Fqmn$ be an $\Fqm$-linear code and $\m_k=n$. Then $C$ is $i$-MRD if and only if $ k-i+2 \leq \m_1(C^\perp) $. In particular, $C$ is $i$-MRD of degree $i-1$ if and only if $\m_1(C^\perp) =k-i+2$.
\end{corollary}
\begin{proof} $C$ is an $i$-MRD $\Fqm$-linear code if and only if $\lambda_{\mathcal{B}}(C)$ is an $i$-MRD $\Fq$-linear code. By Theorem \ref{Mariangel},
 this is equivalent to $\mathcal{M}_1(C^{\perp})= a^{\perp}_{t+1-\lfloor t/m\rfloor m}(\lambda_{\mathcal{B}}(C))\geq 2+ \lfloor \frac{t-r}{m}\rfloor= 2+ \left\lfloor \frac{t-(1+(i-1)m)}{m}\right\rfloor =k-i+2$.
Furthermore, $C$ is an $i$-MRD $\Fqm$-linear code of degree $i-1$ if and only if $\lambda_{\mathcal{B}}(C)$ is also of degree $i-1$, which by Theorem \ref{Mariangel} is
equivalent to $\mathcal{M}_1(C^{\perp})=a^{\perp}_1(\lambda_{\mathcal{B}}(C))=2+t/m-i=k-i+2$.

\end{proof}

\begin{theorem}\label{consequence}
 Let $\mC \leq (\F_q)_{n,m} $ be an $\Fq$-linear code of dimension $t$ and $m \mid t$. The following hold:
\begin{enumerate}
\item Let $a_{1+(\lceil t/m\rceil-1)m}(\mC)=n$. Then
  $\mC$ is an $i$-MRD code of degree $\deg(\mC)=i-1$ if and only if $\mC^\perp$ is an A$^{i-1}$MRD code i.e. $\rdef(\mC^\perp)=i-1$.
\item $a_{1+(\lceil t/m\rceil-1)m}(\mC)<n$ if and only if
$\mC$ is an A$^{t/m}$AMR code i.e. $\rdef(\mC^\perp)=t/m$.
\end{enumerate}
\end{theorem}
\begin{proof}

\begin{enumerate}
\item
By Theorem \ref{Mariangel}, $\mC$ is an $i$-MRD code of degree $i-1$ if and only if $\rdef(\mC^\perp)=\lfloor t/m\rfloor+1-d^\perp=\lfloor t/m\rfloor+1-\lceil t/m\rceil+i-2=i-1.$

\item By Lemma \ref{lemma-mk=n}, we have\\$a_{1+(\lceil t/m\rceil-1)m}(\mC)<n \Leftrightarrow a_1(\mC^\perp)=1 \Leftrightarrow \rdef(\mC^\perp)= \lfloor t/m \rfloor+1-1= t/m.$
\end{enumerate}
\end{proof}

\begin{corollary}\label{consequence-Gab}
 Let $C \leq \Fqmn$ be an $\Fqm$-linear $[n,k,d]$ code. The following hold:
\begin{enumerate}
\item If $\m_k=n$, then $C$ is an $i$-MRD code of degree $\deg(C)=i-1$ if and only if $C^\perp$ is an A$^{i-1}$MRD code i.e. $\rdef(C^\perp)=i-1$.
\item $\m_k<n$ if and only if $C$ is an A$^k$AMR code i.e. $\rdef(C^\perp)=k$.
\end{enumerate}
\end{corollary}

\begin{theorem}\label{Miguel} Let $\mC \leq (\F_q)_{n,m} $ be an $\Fq$-linear code of dimension $t$ and $m\mid t$. The following facts hold:
\begin{enumerate}
\item Assume $a_{1+(\lceil t/m\rceil-1)m}(\mC)=n$ and $\rdef(\mC)\geq 1$. Then $\mC^\perp$ is $\mathrm{AMRD}$ if and only if
$a_{1+m}(\mC)=d+\rdef(\mC)+1$. Therefore, if $\mC$ is AMRD, then $\mC$ is dually AMRD if and only if $a_{1+m}(\mC)=d+2.$
\item If $a_{1+(\lceil t/m\rceil-1)m}(\mC)<n$, then $\mC^{\perp}$ is $\mathrm{AMRD}$ if and only if
$\lceil t/m\rceil=1$. Moreover, $\mC$ is dually AMRD with $t/m=1$ if and only if $a_{1+(\lceil t/m\rceil-1)m}(\mC)<n$ and $d=n-1$.
\end{enumerate}
\end{theorem}
\begin{proof}
\begin{enumerate}
\item Let $\mC^{\perp}$ AMRD. Since $\rdef(\mC^\perp)=1$, then By Theorem \ref{consequence} (1) we have $\deg(\mC)=1$. Therefore $a_{1+m}(\mC)=n-\left\lfloor \frac{t-(1+m)}{m}\right\rfloor=n-\lceil t/m\rceil+2=(n- \lceil t/m\rceil+1-d)+d+1=\rdef(\mC)+d+1.$ Reciprocally, let $a_{1+m}(\mC)=d+\rdef(\mC)+1$. Then $a_{1+m}(\mC)=d+n-\lceil t/m\rceil+1-d+1=n-\lceil t/m\rceil+2=n-\left\lfloor \frac{t-(1+m)}{m}\right\rfloor$. Since $\rdef(\mC) \geq 1$, then $d \neq n-\lceil t/m\rceil+1=n-\lfloor \frac{t-1}{m}\rfloor$ and $\mC$ is a $2$-AMRD with $\deg(\mC)=1$. Therefore $\rdef(\mC^\perp)=1$.
\item The result easily follows from Theorem \ref{consequence} (2).
\end{enumerate}
\end{proof}

\begin{corollary}\label{corolario2} Let $C \leq \Fqmn$ be an $\Fqm$-linear $[n,k,d]$ code.
The following facts hold:
\begin{enumerate}
\item Assume $\mM_k(C)=n$ and $\dc\geq 1$. Then $\dual$ is $\mathrm{AMRD}$ if and only if
$\m_{2}(C)=d+\defect(C)+1$. Therefore, if $C$ is AMRD, then $C$ is dually AMRD if and only if $\m_2(C)=d+2.$
\item If $\mM_k(C)<n$, then $\dual$ is $\mathrm{AMRD}$ if and only if
$k=1$. Moreover, $C$ is a $1$-dimensional dually AMRD code if and only if $\mM_k(C)<n$ and $d=n-1$.
\end{enumerate}
\end{corollary}

\begin{example} {\rm If $\m_k(C)< n$, then $C$ is dually AMRD if and only if $k=1$ and $\m_k=d=n-1$ i.e. $C=\langle v_1,\ldots, v_n\rangle$, where the maximum number of coordinates $v_1,\ldots, v_n$ that are linearly independent over $\F_q$ is $n-1$.  For example, let $q=2$, $n=m=4$ and $\F_{2^4}=\F_2[\alpha]$, where
$\alpha$ satisfies $\alpha^4+\alpha+1=0$. The 1-dimensional code $C \leq
\F_{2^4}^4$ generated by $(1, \alpha,\alpha^2,0)$ has minimum rank $d=3$ and since $(0,0,0,1) \in C^\perp$, then $d^\perp=1$. Therefore $C$ is dually AMRD. Another example is the code $\widehat{\mathcal{G}}$, where $\mathcal{G}$ is the Gabidulin code with $\gen(\mathcal{G})=M_k(v_1,\ldots, v_n)$ (see Theorem \ref{exist-dually-1}).}
\end{example}
\begin{definition} {\rm An $\F_q$-linear code $\mC \leq (\F_q)_{n,m}$ is called a \textit{$2$-AMRD} $\Fq$-linear code if and only if $\mC$ is AMRD and $2$-MRD. Similar definition is given for an $\Fqm$-linear code considering its associated code. }
\end{definition}
In coding theory under the Hamming metric, these $2$-AMRD codes are called near MDS codes (see \cite{Dodunekov, Faldum-Willems}). We can easily see that if $m\mid t$, then the concepts dually AMRD and $2$-AMRD agree, as the following theorem shows.

\begin{theorem} Let $\mC \leq (\F_q)_{n,m}$ be an
$\F_q$-linear code of dimension $t$ with minimum distance $d$ and dual minimum distance $d^\perp$.
The following facts hold:
\begin{enumerate}
\item If $\mC$ is $2$-AMRD, then $a^{\perp}_{1+t-(\lfloor t/m\rfloor-1)m}(\mC)=\lceil t / m\rceil+ 2$ and $a^{\perp}_{1+t-\lfloor t/m\rfloor m}(\mC)=\lceil t /m \rceil$.
\item If $m \mid t$ and $\mC$ is $2$-AMRD, then $\mC^\perp$ is $2$-AMRD.
\item Let $m \mid t$. Then $\mC$ is a $2$-AMRD code if and only if $\mC$ is a dually AMRD code with $t/m>1$.
\end{enumerate}

\end{theorem}

\begin{proof}
\begin{enumerate}
\item Let $\mC$ be a $2$-AMRD code. Then $a_1(\mC)=n-\lceil t/m\rceil$, $a_{1+m}(\mC)=n-\left \lfloor \frac{t-(1+m)}{m} \right\rfloor$ and $a_{1+m}(\mC)=1+a_1(\mC)$. We see that the sequence $$ a_{1+m}< a_{1+2m}< \ldots < a_{1+(\lceil t/m\rceil-1)m}=n $$ is exactly the sequence $$n-\left\lfloor \frac{t-(1+m)}{m} \right\rfloor < n-\left\lfloor \frac{t-(1+2m)}{m} \right\rfloor < \ldots< n-\left\lfloor \frac{t-(1+(\lceil t/m\rceil-1)m)}{m} \right\rfloor=n,$$ which has no gaps. By Theorem \ref{theorem 32}
$$\{n+1-a^{\perp}_{1+t+(n-\lceil \frac{2t+1}{m}\rceil)m}, \ldots, n+1-a^{\perp}_{1+t-\lfloor t/m\rfloor m}\}=[n] \backslash\{ a_1, \ldots, a_{1+(\lceil t/m\rceil-1)m}\}.$$
  Therefore $n+1-a^{\perp}_{1+t-\lfloor t/m\rfloor m}(\mC)=a_{1+m}(\mC)-1=n-\lceil t/m \rceil +2$ and $$n+1-a^{\perp}_{1+t-(\lfloor t/m\rfloor-1)m}(\mC)=a_1(\mC)-1=n- \lceil t/m\rceil-1.$$
Thus $a^{\perp}_{1+t-(\lfloor t/m\rfloor-1)m}(\mC)=\lceil t / m\rceil+ 2$ and $a^{\perp}_{1+t-\lfloor t/m\rfloor m}(\mC)=\lceil t /m \rceil$.
\item In this case $a_1(\mC^\perp)=t/m$ and $a_{1+m}(\mC)=t/m+2$. Therefore $\rdef(\mC^\perp)=1$ and $a_{1+m}(\mC^\perp)=n-\left\lfloor \frac{mn-t-(1+m)}{m} \right\rfloor.$
\item By Theorem \ref{Miguel} (1) we have that $\mC$ is dually AMRD with $t/m >1$ if and only if $\mC$ is AMRD and $a_{1+m}(\mC)=d+2$, which is equivalent to $\mC$ being AMRD and $a_{1+m}(\mC)=n-\lceil t/m \rceil+2=n-\lfloor \frac{t-(1+m)}{m}\rfloor$.

\end{enumerate}
\end{proof}

\begin{corollary} Let $C \leq \Fqmn$ be an $\Fqm$-linear $[n,k]$ code with minimum distance $d$ and dual minimum distance $d^\perp$.
The following facts hold:
\begin{enumerate}
\item If $C$ is $2$-AMRD, then $C^\perp$ is $2$-AMRD.
\item $C$ is a $2$-AMRD code if and only if $C$ is a dually AMRD code with dimension $k>1$.\\
\end{enumerate}

\end{corollary}
\textbf{Acknowledgement.} The author is very grateful to Professor Joachim Rosenthal for his support and hospitality at the University of Zurich.

\end{document}